\documentclass[a4paper,11pt]{article}
\usepackage[margin=3cm]{geometry}
\usepackage{amssymb, amsmath, amsthm}
\usepackage{xspace}
\usepackage{amsmath,amsfonts}
\usepackage[bookmarks,linkcolor=black]{hyperref}
\usepackage[font=small,labelfont=bf]{caption}
\usepackage[font=small,labelfont=normalfont]{subcaption}
\usepackage{graphicx}
\usepackage{enumerate,paralist}
\usepackage{cite}
\usepackage{booktabs}
\usepackage{wrapfig}

\makeatletter
\let\@fnsymbol\@arabic
\makeatother

\title{Colored Non-Crossing Euclidean Steiner Forest}

\author{%
  Sergey~Bereg\thanks{University of Texas, Dallas, USA.} 
  \and 
  Krzysztof~Fleszar\thanks{Lehrstuhl f\"ur Informatik I,
    Universit\"at W\"urzburg, Germany.}
  \and
  Philipp~Kindermann\footnotemark[2]
  \and
  Sergey~Pupyrev\thanks{Department of Computer Science, 
    University of Arizona, Tucson, Arizona, USA.}
    $^{,}$\thanks{Institute of Mathematics
    and Computer Science, Ural Federal University, Russia.}
  \and
  Joachim~Spoerhase\footnotemark[2]
  \and
  Alexander~Wolff\footnotemark[2]}

\newcommand{\cest}[1]{$#1$-CESF\xspace}
\newcommand{\kcest}{\cest{k}}
\newcommand{\eps}{\ensuremath{\varepsilon}\xspace}
\DeclareMathOperator{\opt}{OPT}
\DeclareMathOperator{\alg}{ALG}
\DeclareMathOperator{\bg}{BG}

\usepackage{arydshln}
\theoremstyle{plain}
\newtheorem{theorem}{Theorem} 
\newtheorem{lemma}{Lemma}

\graphicspath{{pics/}}

\begin{document}

\maketitle

\begin{abstract}
  Given a set of $k$-colored points in the plane, we consider the
  problem of finding $k$ trees such that each tree connects all points
  of one color class, no two trees cross, and the total edge length of
  the trees is minimized.  For \mbox{$k=1$}, this is the well-known
  Euclidean Steiner tree problem.  For general~$k$, a
  $k\rho$-approximation algorithm is known, where $\rho \le 1.21$ is
  the Steiner ratio.

  We present a PTAS for $k=2$, a $(5/3+\eps)$-approximation algorithm
  for $k=3$, and two approximation algorithms for general~$k$, with
  ratios $O(\sqrt n \log k)$ and~$k+\eps$.
\end{abstract}

\section{Introduction}

\emph{Steiner tree} is a fundamental problem in combinatorial
optimization.  Given an edge-weighted graph and a set of vertices
called \emph{terminals}, the task is to find a minimum-weight subgraph
that connects the terminals.  For \emph{Steiner forest}, the terminals
are colored, and the desired subgraph must connect, for each color,
the terminals of that color.

In this paper, we consider Steiner forest in a geometric setting and
add to it the requirement of planarity.  More precisely,
we consider the problem of computing, for a $k$-colored
set of points in the plane (which we also call \emph{terminals}), $k$
pairwise non-crossing Euclidean Steiner trees, one 
for each color. Note that such trees exist for every given set of points.
We call the problem of minimizing the total length of
these trees \emph{$k$-Colored Non-Crossing Euclidean Steiner Forest}
(\kcest).
Figure~\ref{fig:bad} shows some interesting instances of \kcest.

\begin{figure}[t]
  \centering
  \begin{subfigure}[b]{.27\textwidth}
  \centering
    \includegraphics[page=1]{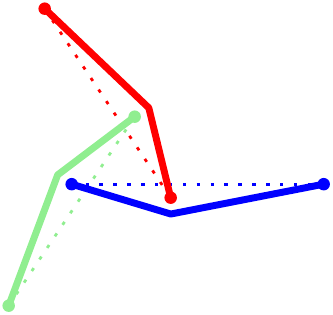}
    \caption{}
    \label{fig:bad-straight}
  \end{subfigure}
  \hfil
  \begin{subfigure}[b]{.27\textwidth}
  \centering
    \includegraphics[page=4]{bad_new}
    \caption{}
    \label{fig:bad-length}
  \end{subfigure}
  \hfil
  \begin{subfigure}[b]{.43\textwidth}
  \centering
    \includegraphics[page=3]{bad_new}
    \caption{}
    \label{fig:bad-maze}
  \end{subfigure}
  \caption{Difficult examples for \kcest. (a)~The optimum
    contains no straight-line edge. (b)~Segment $ab$ is used
    twice by the black curve.  (c)~In the optimum, the black curve can
    be made arbitrarily longer than the corresponding
    straight-line segment (if every gray segment represents a
    different color).}
  \label{fig:bad}
\end{figure}

The problem is
motivated by a method that Efrat et al.~\cite{EHKP14} suggested
recently for visualizing embedded and clustered graphs.  They
visualize clusters by regions in the plane that enclose related graph
vertices.  Their method attempts to reduce visual clutter and optimize
``convexity'' of the resulting regions by reducing the amount of
``ink'' necessary to connect all elements of a cluster.  Efrat et
al.~\cite{EHKP14} proposed the problem \kcest and provided a simple 
$k\rho$-approximation algorithm, where
$\rho$ is the \emph{Steiner ratio}, that is, the supremum, over all
finite point sets in the plane, of the ratio of the
total edge length of a minimum spanning tree over the total edge length
of a Euclidean Steiner tree (EST). 
Chung and Graham~\cite{CG85} showed that $\rho \le 1.21$.

\paragraph{Our contribution.}
The middle column of Table~\ref{tab:results} shows our results. 
For \cest{k}, we
present a deterministic $(k+\eps)$- 
and a randomized $O(\sqrt{n}\log k)$-appro\-xi\-ma\-tion algorithm;
see Section~\ref{sec:kcest}. The main result of our paper is that
\cest{2} admits a polynomial-time approximation scheme (PTAS); see Section~\ref{sec:2cest}. 
By a non-trivial modification of the PTAS, we prove that \cest{3} admits a
$(5/3+\eps)$-approximation algorithm; see Section~\ref{sec:3cest}.

Our PTAS for~\cest{2} uses some ideas of Arora's
algorithm~\cite{Aro98} for EST, which is equivalent to \cest{1}.  Since,
in a solution to~\cest{2}, the two trees are not allowed to cross, our
approach differs from Arora's algorithm in several respects.  We use a
different notion of \emph{$r$-lightness}, and by a
\emph{portal-crossing reduction} we achieve that each portal is
crossed at most three times.  More care is also needed in the
perturbation step and in the base case of the dynamic program.

\begin{table}[htb]
  \centering
  \begin{tabular}{@{}c@{\qquad}c@{\qquad}c@{}}
    \toprule
    $k$ & \kcest & planar graph \\ \midrule
    1 & EST: NP-hard \cite{GJ79}, & ST: NP-hard \cite{GJ79}, \\
     & $1+\eps$ \cite{Aro98,Mit99a} & $1+\eps$ \cite{BKM09} \\
    \hdashline
    2 & $1+\eps$ (Theorem~\ref{thm:2trees}) \\
    3 & $5/3+\eps$ (Theorem~\ref{thm:3cest}) \\
    \hdashline
    general $k$ & $k+\eps$ (Theorem~\ref{thm:kcest-k+eps}), &
    $k$ const.-size nets on 2 faces, exact~\cite{KMN01} \\
     & $O(\sqrt{n}\log k)$ (Theorem~\ref{thm:kcest-sqrtnlogk})\\
    \hdashline
    $n/2$ & NP-hard \cite{BF98}, &
    $k$ size-2 nets on $h$ faces, exact \cite{EN11} \\ 
     & $O(\sqrt{n}\log n)$ \cite{CHKL13} \\
    \bottomrule
  \end{tabular}
  \caption{Known and new results for \kcest (hardness and
    approximation ratios).}
  \label{tab:results}
\end{table}

\paragraph{Related Work.}

Apart from the result of Efrat et al.~\cite{EHKP14}, so far the only
two variants of \kcest that have been studied are those with extreme
values of~$k$.  As mentioned above, \cest{1} is the same as EST, which is 
NP-hard \cite{GJ79}.  Arora~\cite{Aro98} and Mitchell~\cite{Mit99a}
showed independently that EST admits a PTAS.  The other extreme value
of~$k$, for which
\kcest has been considered, is $k=n/2$.  This is the problem of
joining specified pairs of points via non-crossing curves of minimum
total length.  Liebling et al.~\cite{LMMPS95} gave some heuristics for
this problem.  Bastert and Fekete~\cite{BF98} claimed that \cest{(n/2)}
is NP-hard, but their proof has not been formally published.
Recently, Chan et al.~\cite{CHKL13} considered \cest{(n/2)} in the
context of embedding planar graphs at fixed vertex locations.  They
gave an $O(\sqrt n \log n)$-approximation algorithm based on an idea
of Liebling et al.~\cite{LMMPS95} for computing a short non-crossing
tour of all given points.

Several set visualization techniques assume the same setting
where the input is a multicolored point
set~\cite{ARRC11,HKKLS+13,CPC09,RBKMSW08}.
These techniques allow regions to cross,
while the regions that correspond to (the geodesic hulls of) the trees
in our approach are non-crossing.  Reinbacher et al.~\cite{RBKMSW08}
consider the problem of computing a minimum-perimeter red polygon
inside a given red polygon that contains red and blue points.

There is substantial work on the case where there are obstacles
in the plane.  Note that, in contrast to \kcest, a valid
solution may not exist in that setting.  For a single color (that is,
\cest{1} with obstacles), M{\"u}ller-Hannemann and Tazari~\cite{MT10}
give a PTAS.
The same variant is considered by Razaghpour~\cite{Raz08}.
Papadopoulou~\cite{Pap99} gave
an algorithm for finding minimum-length non-crossing paths joining
pairs of points (that is, \cest{n/2}) on the boundary of a single
polygon.
Erickson and Nayyeri~\cite{EN11} generalize the problem to
the case of points on the boundaries of $h$ polygonal obstacles; their
algorithm is exponential in~$h$.
A practical aspect of the
problem---computing non-crossing paths of specified thickness---was
studied by Polishchuk and Mitchell~\cite{PM07}.  Their algorithm
computes a representation of the thick paths inside a simple polygon;
they also show how to find shortest thick disjoint paths joining
endpoints on the boundaries of polygonal obstacles (with exponential
dependence on the number of obstacles).
They also prove that the problem is hard to approximate.
The main difficulty with multiple obstacles is deciding which homotopy
class of the paths gives the minimum length.  If the homotopy class of
the paths is specified, then the problem is significantly
easier~\cite{Bes03,EKL06,v-hcor-GD12}.

The graph version of the problem has been studied in the
context of VLSI design. Given an edge-weighted plane graph~$G$ and a
family of $k$ vertex sets (called nets), the goal is to find a set of
$k$ non-crossing Steiner trees interconnecting the nets such that the
total weight is minimized.  The problem is clearly NP-hard, as the
special case $k=1$ is the graph Steiner tree problem (ST), which is
known to be NP-hard~\cite{GJ79}.  ST admits a PTAS~\cite{BKM09}.
On planar graphs, \kcest can be solved in $O(2^{O(h^2)}n \log k)$
time~\cite{EN11} for $k$ terminal pairs (that is, size-2 nets) if all
terminals lie on~$h$ faces of the given $n$-vertex graph
and in $O(n \log n)$ time for $h=2$ and $k$ constant-size
nets~\cite{KMN01}.  We list these results in Table~\ref{tab:results};
many entries are still open.

A restricted version of the problem in which all nets consist of
exactly two terminals (that is, $k=n/2$) was considered by Takahashi
et al.~\cite{TSN96}. They show that, if $G$ satisfies the so-called
2-face condition, that is, if all terminals lie on at most two of the
face boundaries of $G$, then such a set of $k$ non-crossing paths can
be found in time $O(n \log n)$.  Subsequently, Kusakari et
al.~\cite{KMN01} extended this result by providing an optimal $O(n
\log n)$-time algorithm for the variant without the two-terminal
restriction (but still requiring the 2-face condition).  The general
case, where terminals lie on at most $h$ face boundaries, was studied
by Erickson and Nayyeri~\cite{EN11}.  They presented an algorithm that
connects the given $k$ terminal pairs by non-crossing paths of minimum
total length in $O(2^{O(h^2)}n \log k)$ time.  The complexity of the
problem on planar graphs for arbitrary $h>2$ seems to be open.

In the \emph{group Steiner tree} problem, one is given a $k$-colored point set
and the task is to find a minimum-length tree that connects at least
one point of each color.  The problem is discussed in a survey by
Mitchell~\cite{m-gspno-HBCG00}.
Another related problem is that of constructing a minimum-length
non-crossing path through a given sequence of points in the plane.
Its complexity status remains open~\cite{PB05,Lof11}

\section{Algorithms for \kcest}
\label{sec:kcest}

Despite its simple formulation, the \kcest problem seems to be rather
difficult.  There are instances where the optimum contains no
straight-line edges or contains paths with repeated
line segments; see Figures~\ref{fig:bad-straight}
and~\ref{fig:bad-length}.  This shows that obvious greedy
algorithms fail to find an optimal solution, as Liebling et
al.~\cite{LMMPS95} observed.  They also provided an instance of the
problem in a unit square for $k=n/2$ in which the length of an optimal solution is in
$\Omega(n \sqrt{n})$, whereas the trivial lower bound (the sum of
lengths of straight-line segments connecting the pairs of terminals) is
only $O(n)$. The example is based on the existence of expander graphs
with a quadratic number of edge crossings. In Figure~\ref{fig:bad-maze}, we
provide an example in which the length of one of the curves in the optimal
solution can be arbitrarily bigger than the trivial lower bound
for the corresponding color.

Efrat et al.~\cite{EHKP14} suggested an approximation algorithm for
\kcest.  The key ingredient of their algorithm is the following
observation, which shows how to make a pair of given trees
non-crossing: reroute one of the trees using a ``shell'' around the
other tree.  For any geometric graph $G$, we denote its total edge
length by~$|G|$.

\begin{lemma}[Efrat et al.~\cite{EHKP14}]
\label{lm:shell}
Let $R$ and $B$ be two trees in the plane spanning red and blue terminals, respectively.
Then, there exists a tree $R'$ spanning the red terminals such that
\begin{enumerate}[(i)]
	\item $R'$ and $B$ are non-crossing and
  \item $|R'| \le |R| + 2|B|$.
\end{enumerate}
\end{lemma}

Efrat et al.~\cite{EHKP14} start with $k$ (possibly intersecting)
minimum spanning trees, one for each color.  Then, they iteratively go
through these trees in order of increasing length.  In every step,
they reroute the next tree by laying a shell around the current
solution as in Lemma~\ref{lm:shell}.  Their algorithm has
approximation factor~$k\rho$.  We now show that the algorithm even
yields approximation factor~$k+\eps$ if we use a PTAS for EST for the
initial solution to each color.

\begin{theorem}
  \label{thm:kcest-k+eps}
  For every $\eps>0$, there is a $(k+\eps)$-approximation algorithm for
  \cest{k}.
\end{theorem}

\begin{proof}
  Fix an optimal solution~$\mathcal{T}$.  For $i=1,\dots,k$, let $T_i$
  be the length of the Steiner tree spanning color~$i$
  in~$\mathcal{T}$.  Hence, $\opt = \sum_{i=1}^k T_i$.  For each
  color~$i$, use a PTAS to compute a Steiner tree of length~$P_i\le
  (1+\eps)T_i$. Now, consider the trees one by one in non-decreasing
  order of their lengths. W.l.o.g., assume that~$P_1\le\ldots\le P_k$.

  Denote by~$S_i$ the length of the Steiner tree spanning color~$i$ in
  our solution~$\cal S$.  Then Lemma~\ref{lm:shell} yields $S_1=P_1$
  and $S_2 \le P_2+2P_1$. For the $i$-th tree, we add a shell around
  every tree~$j<i$. Observe that the tree and the shell give us a
  Steiner tree for~$i$ that does not intersect any tree built in
  step~$j<i$. Thus, we get~$S_i\le{P_i + \sum_{j=1}^{i-1} 2P_j}$
  and $|{\cal S}| = \sum_{i=1}^k \big(P_i + \sum_{j=1}^{i-1} 2P_j \big)$.
  Let $\overline{P}=(\sum_{i=1}^k P_i)/k$ be the average length of the Steiner
  trees. Since the Steiner trees~$P_i$ are ordered by non-decreasing
  lengths, it holds for all~$i$ that $\sum_{j=1}^{i} P_j \le
  i\cdot\overline{P}$.  This yields
  \[S_i=P_i+\sum_{j=1}^{i-1} 2P_j=\sum_{j=1}^{i-1} P_j+\sum_{j=1}^{i}
  P_j\le (2i-1)\overline{P}~,\] which sums up to
  \[|{\cal S}| \le \sum_{i=1}^k (2i-1)\overline{P} =
  k^2\cdot\overline{P}~.\]
	
  Since $k\cdot\overline{P} = \sum_i P_i\le (1+\eps) \opt$, we have
  $\overline{P}\le (1+\eps)\opt/k$.  Hence, the total length of our
  solution is 
  \[|{\cal S}| \le k^2\cdot\overline{P} \le (1+\eps) k \opt~.\qedhere\]
\end{proof}

For even~$k$, we can slightly improve on this by using our PTAS for \cest{2}
(Theorem~\ref{thm:2trees}).

\begin{theorem}
  \label{thm:kcest-k+eps-improved}
  For every $\eps>0$, there is a $(k-1+\eps)$-approximation algorithm
  for \cest{k} if $k$ is even.
\end{theorem}

\begin{proof}
  Fix an optimal solution~$\mathcal T$.  For
  $1\le i\le k$, let~$P_i$ be the set of terminals of color~$i$.  Call
  the terminals in $Q_1=\bigcup_{i=1}^{k/2}P_i$ red and those in
  $Q_2=\bigcup_{j=1+k/2}^{k}P_j$ blue.  Let~$\mathcal T^*$ be an
  optimal solution to the resulting \cest{2} instance~$I^*$. Obviously,
  $|\mathcal T^*|\le |\mathcal T|$.

  We use the PTAS for \cest{2} of Theorem~\ref{thm:2trees} to compute a
  solution~$\mathcal S^*$ to~$I^*$ with $|\mathcal
  S^*|\le(1+\eps')|\mathcal T^*|$, where $\eps'=\eps/(k-1)$.
  In general, $\mathcal S^*$ is not a valid solution to~$I$.
  Let~$S_1$ and~$S_2$ be the trees connecting~$Q_1$ and~$Q_2$ in~$\mathcal S^*$,
  respectively. We now create, for $1\le i\le k/2$, a tree~$R_i$
  connecting the terminals in~$P_i$, as follows.
  Let~$R_1$ be the smallest subtree of~$S_1$ spanning the terminals in~$P_1$.
  Thus, $|R_1|\le|S_1|$.  Then,
  we create~$R_2$ by laying a shell around~$R_1$ that goes through all
  terminals in~$P_2$. Note that, at this point, $R_2$ still contains a cycle
  that has~$R_1$ in its interior. We iteratively create~$R_i$,
  $3\le i \le k/2$, by laying a shell around the outer boundary of~$R_{i-1}$
  that goes through all terminals in~$P_i$. Finally, we cut the cycles
  of~$R_2,\dots,R_{k/2}$ at some point to create trees. Since each of
  these
  trees consists of a shell around~$S_1$, it holds that $|R_i|\le 2|S_1|$,
  $2\le i\le k/2$.  Analogously, we compute $R_{1+k/2},\dots,R_k$ with
  $|R_{1+k/2}|\le |S_2|$ and $|R_j|\le 2|S_2|$ for $2+k/2\le j\le k$.

  Our solution~$\mathcal{R}$ to~$I$ consists of~$R_1,\dots,R_k$; its
  total length is 
  \begin{align*}
    |\mathcal{R}|
    &  =   |R_1|+\sum_{i=2}^{k/2}|R_i|+|R_{1+k/2}|+\sum_{j=2+k/2}^{k}|R_j|\\
    & \le  |S_1|+(k/2-1)\cdot 2|S_1|+|S_2|+(k/2-1)\cdot 2|S_2|\\
    &  =   (k-1)|S_1| + (k-1)|S_2| = (k-1)|\mathcal S^*|\\
    & \le  (k-1)(1+\eps')|\mathcal T^*|
    \le (k-1+(k-1)\eps')|\mathcal T|\\
    & = (k-1+\eps)|\mathcal T|~.
    \qedhere
  \end{align*}
\end{proof}  

Next, we present an approximation algorithm for \cest{k} whose ratio
depends only logarithmically on~$k$, but also depends on~$\sqrt{n}$.
The algorithm employs a space-filling curve through a set of given
points.  The curve was utilized in a heuristic for \cest{(n/2)} by
Liebling et al.~\cite{LMMPS95}.  Recently, Chan et al.~\cite{CHKL13}
showed that the approach yields an $O(\sqrt n \log n)$-approximation for
\cest{(n/2)}.  We show that similar arguments yield approximation
ratio $O(\sqrt n \log k)$ for general~$k$.

\begin{theorem}
  \label{thm:kcest-sqrtnlogk}
  \cest{k} admits a (randomized) $O(\sqrt n \log k)$-approximation
  algorithm.
\end{theorem}
\begin{proofWithWrapfigMath}
  Chan et al.~\cite{CHKL13} gave a randomized algorithm to construct a
  curve~$C$ through
  the given set~$P$ of $n$ points.  Their curve has
  small \emph{stretch}, that is, the ratio between the Euclidean
  distance $d(p,q)$ of two points $p,q\in P$ and their distance
  $d_C(p,q)$ along the curve is small.  Assuming that the points are
  scaled to lie in a unit square, Chan
  et al.\ showed, for a fixed pair of points $p,q\in P$,
  \[\mathbb{E}[d_C(p,q)] \le O(\sqrt{n} \log(\frac{1}{d(p,q)})) \cdot d(p,q)~.\]
  Using~$C$, we construct a solution to \cest{k} so that, for every
  color, the terminals are visited in the order given by the curve;
  and thus, the solution to every color is a path. All paths can be
  wrapped around the curve without intersecting each other; see
  Figure~\ref{fig:wrap}.
  
  \begin{figure}[tb]
    \centering
    \begin{subfigure}[b]{.47\textwidth}
      \centering
      \includegraphics[scale=1.5,page=2]{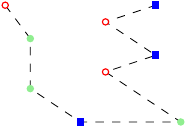}
      \caption{}
    \end{subfigure}
    \hfil
    \begin{subfigure}[b]{.47\textwidth}
      \centering
      \includegraphics[scale=1.5,page=3]{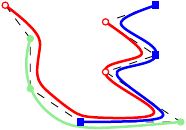}
      \caption{}
    \end{subfigure}
  \caption{(a)~A low-stretch curve $C$ through the terminals;
    (b)~a \cest{3} solution to the instance created by wrapping paths
    around $C$.}
    \label{fig:wrap}
  \end{figure}

If the order of the points along the curve for a specific color~$i$ is
$p_1^i, \dots, p_{n_i}^i$,
then the length of the corresponding path is
\[\sum_{j=1}^{n_i-1} d_C(p_j^i, p_{j+1}^i) = d_C(p_1^i, p_{n_i}^i)~.\]
Let $\overline{d}=\sum_{i=1}^k d(p_1^i, p_{n_i}^i) / k$.
The total (expected) length of the solution is
\begin{flalign*}
\alg = \sum_{i=1}^k \mathbb{E}[d_C(p_1^i, p_{n_i}^i)] &\le
       \sum_{i=1}^k O(\sqrt{n} \log(1/d(p_1^i, p_{n_i}^i))) \cdot d(p_1^i, p_{n_i}^i)\\
      &\le \sum_{i=1}^k O(\sqrt{n} \log(1/\overline{d})) \cdot \overline{d}~.
\end{flalign*}
Since the optimal solution to~$P$ connects all pairs of terminals of the same color (possibly using
non-straight-line curves), $\opt \ge \sum_{i=1}^k d(p_1^i, p_{n_i}^i)=k\overline{d}$.
Hence,
\[\alg \le \sum_{i=1}^k O(\sqrt{n} \log(k/\opt)) \cdot \opt/k \le O(\sqrt{n} \log k) \opt~.\tag*{\qed}\]
\end{proofWithWrapfigMath}

\section{PTAS for \cest{2}}\label{sec:2cest}

In this section, we show that \cest{2} admits a PTAS.
We follow the Arora's approach for computing EST~\cite{Aro98}, which
consists of the following steps. First, Arora performs a recursive 
geometric partitioning of the plane using a quadtree and snaps the input 
points to the corners of the tree. Next, he defines an \emph{$r$-light} solution,
which is allowed to cross an edge of a square in the quadtree at most~$r$ 
times and only at so-called \emph{portals}. Then he builds an optimal \emph{portal-respecting}
solution using dynamic programming, and finally trims the 
edges of the solution to get the result. To get an algorithm for \cest{2}, 
we modify these steps as follows:
\begin{compactenum}[(i)]
\item The perturbation step, which snaps the terminal to a grid, is modified
to avoid crossings between trees. Similarly, the reverse step transforming a perturbed
instance to an original one is  different; see Lemmas~\ref{lem:perturbed}
and \ref{lem:unperturbed}.

\item We use a different notion of an $r$-light solution in which every \emph{portal} 
is crossed at most~$r$ times. We devise a
\emph{portal-crossing reduction} that reduces the number of crossings
to $r=3$; see Lemma~\ref{lem:3light}.

\item The base case of dynamic programming needs a special modification;
it computes a set of crossing-free Steiner trees of minimum total 
length (see Lemma~\ref{lem:dp-ptas}).
\end{compactenum}

We assume that the bounding
rectangles of the two sets of input terminals overlap; otherwise, we can use
a PTAS for the Steiner tree of each input set individually.
We first snap the instance to an $(L\times L)$-grid with~$L = O(n)$.
We proceed as follows.  Let~$L_0$ be the diameter of the smallest
bounding box of the given \cest{2} instance.
We place an $(L\times L)$-grid of granularity $g=L_0/L$ inside the bounding box.
By scaling the instance appropriately, we can assume that the granularity
is $g=1$. We move each terminal of one color to the nearest grid point in an
even row and column, and each terminal of the other color to the nearest 
grid points in an odd row and column. Thus, the grid point for each terminal
 is uniquely defined,
and no terminals of different color end up at the same location. If there are
more terminals of the same color on a grid point, we remove all but one of
them. We call the resulting instance a \emph{perturbed} instance.

  \begin{lemma}
    \label{lem:perturbed}
    Let $\opt_I$ be the length of an optimal solution to a \cest{2} instance~$
I$
    of~$n$ terminals and let $\eps>0$. 
    There is an $(L\times L)$-grid with $L = O(n/\eps)$
    such that $\opt_{I^*}\le(1+\eps)\opt_I$, where $\opt_{I^*}$ is the length
    of an optimal solution for the perturbed instance~$I^*$.
  \end{lemma}
  \begin{proof}
  Choose $L$ to be a power of 2 within the interval $[3\sqrt{2}n/\eps,6\sqrt{2
}n/\eps]$ and perturb the instance as described above. Consider an optimal 
solution to~$I$ and
  connect each terminal of~$I^*$ to the tree of its respective color.
  For every terminal of~$I^*$, proceed as follows:  Connect the terminal to 
  the closest point of the same color lying on the optimal solution.  If this 
  line segment crosses the tree of the other color, then reroute this tree 
  around the line segment by using two copies of the line segment.
  Two copies suffice even if the other tree is crossed more than once since
  all crossing edges can be connected to the two new line segments.
  The distance between the terminal and the tree
  is at most the distance between the
  terminal and the corresponding terminal in~$I$,
  which is bounded by~$\sqrt{2}$ as we are assuming the unit grid.  
  Hence, we pay at most~$3\sqrt{2}$ for connecting the terminal.
  Since the bounding rectangles of the input terminals overlap, $\opt_I\ge L$.
  Thus, the additional length of an optimal solution to~$I^*$ is
  \[\opt_{I^*} - \opt_I
     \leq 3 \sqrt{2} n
  \le \eps\cdot\opt_I~.\qedhere
  \]
  \end{proof}

The next lemma, proven analogously to Lemma~\ref{lem:perturbed},
 shows that we can transform a solution to the perturbed instance into one to the original instance.
  \begin{lemma}
    \label{lem:unperturbed}
    Given a solution $\mathcal T$ to the perturbed instance as
    defined in Lemma~\ref{lem:perturbed}, we can transform~$\mathcal T$ into a
    solution to the original instance, increasing its
    length by at most $\eps\cdot \opt_I$.
  \end{lemma}

   In the following, we assume that the instance is perturbed.
   We place a quadtree in dependence of two integers $a,b\in[0, \dots, L-1]$ that we choose independently uniformly at random.
   We place the origin of the coordinate system on the bottom left corner of the bounding box of our instance.
   Then we take a box~$B$ whose width and height is twice the width and height of the bounding box.
   We place it such that its bottom left corner has coordinates $(-a,-b)$.
   Note that the bounding box is inside $B$.
   We extend the $(L\times L)$-grid to cover~$B$.  Thus, we have an
   $(L'\times L')$-grid with $L'=2L$. 

   Then we partition $B$ with a quadtree along the $(L'\times
   L')$-grid.  The partition is stopped when the current quadtree box
   coincides with a grid cell. 
   We define the \emph{level of a quadtree square} to be its depth in the quadtree.
   Thus, $B$ has level~$0$, whereas the level of a leaf is bounded by $\log L'=\log (2L)=O(\log n)$.
   Then, for each grid line $\ell$, we define its \emph{level}  as the highest (that is, of minimum value) level of all the quadtree squares that touch $\ell$ (but which are not crossed by it).

   Let $m=\lceil 4 \log L' / \eps \rceil$.
   On each grid line $\ell$ of level $i$, we place $2^i \cdot m$ equally spaced points.
   We call these points \emph{portals}.
   Thus, each square contains at most $m$ portals on each of its edges.
   A solution that crosses the grid lines only at portals is called \emph{portal-respecting}. We show that
   there is a close-to-optimal portal-respecting solution.
   Note that, in contrast to Arora, we first make the solution portal-respecting before reducing the number of crossings on each grid line.
  The proof of the following lemma is similar to the Arora's prove.

\begin{lemma}\label{lem:portalred}
   Let~$\opt_I$ be the length of an optimal solution to a \cest{2}
   instance~$I$, and let $\eps>0$ be the same as defined for $m$.  
   There exists
   a position of the quadtree and a portal-respecting
   solution to~$I$ of length at most $(1+\eps)\opt_I$.
\end{lemma}

\begin{proof}
  Fix an optimal solution $\mathcal T$.
	Move all crossings on the grid lines to the
  closest portals by adding a line segment on each side of the grid.
	Note that the modified solution remains crossing-free.

  Consider a grid line $\ell$ that has a non-empty intersection with the bounding box $B$.
  Let $t(\ell)$ be the number of crossing points between $\ell$ and $\mathcal T$.
  If $i$ is the level of $\ell$, then the inter-portal distance on $\ell$ is $L'/(2^i \cdot m)$.
  Since the position of the quadtree has been chosen uniformly at random,
  the probability
  that the level of $\ell$ is $i$ is at most $2^{i} / L'$.
  Thus, the expected length increase of $\mathcal T$ due to moving the crossings to the portals of $\ell$ is at most
  \[\sum_{i=1}^{\log L'} {\frac{2^i}{L'}} \cdot t(\ell) \cdot {\frac{L'}{2^i \cdot m}} \le {\frac{t(\ell) \log L'}{m}} \le \eps {\frac{t(\ell)}{4}}~,\]
  where the last inequality follows from $m\ge 4 \log L' / \eps$.
  Thus, the expected total length increase is at most $(\eps/4) \sum_{\mathrm{gridline}\,\ell} t(\ell)$.

  It holds that $\sum_{\textrm{gridline\,}\ell} t(\ell) \le 4 \opt$.
  To see this, consider any line segment of the solution. Let $l$ be the length of the line segment.
  Given the granularity $g=1$ of the grid, the line segment crosses at most $l+1$ horizontal grid lines and at most
  $l+1$ vertical grid lines; hence, its contribution to the left-hand side of the equation is at most $4 \cdot l$.

  Thus, we have shown that the expected length increase is at most  $\eps \opt$.
  But then, there exists a position of the quadtree for which the
  total length increase is bounded by $\eps \opt$. 
  We can try out all positions of the quadtree which increases the total running time by the factor of $O(n^2)$.
\end{proof}

  The last ingredient for our dynamic programming is to reduce the number of crossings
 in every portal. We call a solution~\emph{$r$-light}
    if each portal is crossed at most $r$ times.
    (Note that we use a different definition than Arora.
    He defined a solution to be $r$-light if a grid line is crossed at most $r$ times.)

  In the following, we explain an operation which we
  call a \emph{portal-crossing reduction}. We are given a
  portal-respecting solution
  consisting of two Steiner trees $R$ and $B$ (red and blue) and we want to
  reduce (that is, modify without increasing its length) it such that
  $R$ and $B$ pass through each portal at most three times in total.

  \begin{lemma}\label{lem:3light}
    Every portal-respecting solution of \cest{2} can be transformed
    into a 3-light portal-respecting solution without increasing its length.
  \end{lemma}
  \begin{proof}
\begin{figure}[t]
    \centering
    \begin{subfigure}[b]{.23\textwidth}
      \centering
      \includegraphics[page=1]{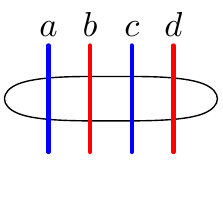}
      \caption{}
      \label{fig:portal-a}
    \end{subfigure}
    \hfil
    \begin{subfigure}[b]{.23\textwidth}
      \centering
      \includegraphics[page=2]{portal_new}
      \caption{}
      \label{fig:portal-b}
    \end{subfigure}
    \hfil
    \begin{subfigure}[b]{.23\textwidth}
      \centering
      \includegraphics[page=3]{portal_new}
      \caption{}
      \label{fig:portal-c}
    \end{subfigure}
    \hfil
    \begin{subfigure}[b]{.23\textwidth}
      \centering
      \includegraphics[page=4]{portal_new}
      \caption{}
      \label{fig:portal-d}
    \end{subfigure}
  \caption{A portal modification for four passes.}
  \label{fig:portal}
\end{figure}
  Consider a sequence of passes through a portal. We assume that there are no terminals in the portals. 
  If two adjacent passes belong to the same tree, then we can eliminate one of them
  by snapping it to the other one. Note that this may create cycles, but they can
  be broken by removing the longest part of each cycle. Therefore, we can assume that the passes
  form an alternating sequence. It suffices
  to show that any alternating sequence of four passes
  can be reduced to two passes by shortening the trees.
  Let $a,b,c$, and $d$ be such a sequence as shown in Figure~\ref{fig:portal-a}.
  We cut the passes $b$ and $c$. This results in two connected components in each tree.
  W.l.o.g., $a$ and the upper part of $c$ belong to the same connected component; see
  Figure~\ref{fig:portal-b}. Otherwise, we can change the colors because
  \begin{inparaenum}[(i)]
  \item $a$ and the lower part of $c$ are connected, and
  \item the upper part of $b$ and $d$ are connected.
  \end{inparaenum}

  Since $R$ and $B$ are disjoint, $d$ and the lower part of $b$
  are in the same connected component; see Figure~\ref{fig:portal-c}. Then, we connect the component as shown in
  Figure~\ref{fig:portal-d} and shorten the trees (e.g., the lower part of $b$ can be
  reduced to a terminal of $R$). Note that the passes~$a$ and~$d$ remain in the
  solution, while the passes~$b$ and~$c$ are eliminated. We repeat the procedure
  for the remaining passes, until there are at most three passes left.
  The length of the solution does not increase because the portal has zero width.
\end{proof}

With the next Lemma~\ref{lem:dp-ptas}, we show how to find a close-to-optimal $3$-light 
portal-respecting solution to the perturbed instance. We assume that an
appropriate quadtree (as defined in Lemma~\ref{lem:portalred}) is given.

\begin{lemma}\label{lem:dp-ptas}
    Let $I$ be an instance of \cest{2} with $n$ terminals, and $\opt_{I}$ be the length
    of an optimal 3-light portal-respecting solution to $I$.
    For every $\eps>0$, there is an algorithm finding
    a solution to $I$ of length at most $(1+\eps)\opt_{I}$
    in time $O(n^{O(1/\eps)})$ using $O(n^{O(1/\eps)})$ space.
  \end{lemma}
  \begin{proof}
    We use dynamic programming with
    a subproblem consisting of
    \begin{enumerate}[(a)]
      \item\label{dyn-quad} a square of the quadtree,
      \item\label{dyn-sequence} a sequence of up to three red and blue points
        on each portal on the border of the square, and
      \item\label{dyn-partition} a \emph{non-crossing} partition of these
        points into sets of the same color.
    \end{enumerate}
     A partition of these points
    is non-crossing if for no four points $a,b,c,d$, occurring in that order on the boundary of the square, it holds that $a$ and $c$ belong to one set of the partition, and $b$ and $d$ to another one.
    The goal is to find an optimal collection of crossing-free red and blue
    Steiner trees, such that each set of the partition and each terminal inside the square is
    contained in a tree of the same color.

    The base case of dynamic programming is a unit square, which is either empty or contains terminals only at corners of the square.	If the square is empty, we consider each set of the partition as an instance of \cest{1} and solve it by the PTAS for EST~\cite{Aro98}.
	For each point set, we force its Steiner tree to lie inside its convex hull, by projecting any part of the solution
	outside the convex hull to its border.
	Since the partition is non-crossing, the convex hulls of its point sets are pairwise disjoint.
	Therefore, the Steiner trees and their union is also a close-to-optimal solution to the base case.
	If the square contains (constantly many) terminals at the
    corners, these terminals are treated in the same way as portals.

	For composite squares in the quadtree, we proceed as follows.
	For the four squares that subdivide the composite square, we consider all 
  combinations of all possible sequences~(\ref{dyn-sequence}) and 
  partitions~(\ref{dyn-partition}) that match together and match the subproblem. 
  In dynamic programming, we already have computed a close-to-optimal solution 
  to every choice of sequence~(\ref{dyn-sequence}) and 
  partition~(\ref{dyn-partition}) of each of the four squares; taking the best 
  combination gives a close-to-optimal solution.

	The size of the dynamic programming table is proportional to the number of 
  subproblems, that is, 
  (\ref{dyn-quad})$\times$(\ref{dyn-sequence})$\times$(\ref{dyn-partition}).
	There are $O(n^2)$ squares~(\ref{dyn-quad}) in the quadtree in total.
	Each square contains $O(\log n/\eps)$ portals.  For each portal, there is a 
  constant number of possible sequences of up to three colored points. Thus, 
  there are 
	$2^{O(\log n/\eps)} = n^{O(1/\eps)}$ possible sequences~(\ref{dyn-sequence}).
	Since the number of non-crossing partitions of a set of $k$ elements is the 
  $k$'th Catalan number~$C_k$, we have 
  $C_{O(\log n/\eps)}<2^{O(\log n/\eps)}=n^{O(1/\eps)}$ possible 
  partitions~(\ref{dyn-partition}).
	In total, we consider $n^{O(1/\eps)}$ subproblems in the dynamic programming.
	
	The running time to solve the base case is
	polynomial in $m=O(\log n/\epsilon)$.  The running time to handle a composite square is polynomial in $(n^{O(1/\epsilon)})^4$, which is $n^{O(1/\epsilon)}$.  Thus, the total running time is bounded by $n^{O(1/\epsilon)}$.
     \end{proof}

Now we prove the main result of the section.

\begin{theorem}
\label{thm:2trees}
  \cest{2} admits a PTAS.
\end{theorem}

\begin{proof}
  Consider an instance $I$ of \cest{2}, let $\opt$ be the length of the optimum, 
  and choose $\eps>0$. By Lemmas~\ref{lem:perturbed}, \ref{lem:portalred} 
  and~\ref{lem:3light}, the length, $\opt'$, of an optimal $3$-light 
  portal-respecting solution to the perturbed version of $I$ is a 
  most $(1+\eps)\opt$. Using Lemma~\ref{lem:dp-ptas}, we find a $3$-light 
  portal-respecting solution to the perturbed instance of length at most 
  \[(1+\eps)\opt' \le (1+\eps)(1+\eps)\opt~.\]
  By Lemma~\ref{lem:unperturbed}, we transform the solution into a solution 
  to $I$ by increasing its length by at most $\eps\cdot\opt$. Therefore, for
  every $\eps'>0$, we can construct a solution to~$I$ of length 
  \[(1+\eps)(1+\eps)\opt + \eps\cdot\opt \le (1+\eps')\opt\]
  by choosing $\eps>0$ appropriately.
\end{proof}

\section{Algorithm for \cest{3}}\label{sec:3cest}

The approach for \cest{2} described in Section~\ref{sec:2cest}
cannot be directly applied to \cest{3} since optimal trees may need to pass
portals many times.  For example, the three paths crossing the portal in
Fig.~\ref{fig:spiral} are difficult because 
we cannot locally reroute to make them $O(1)$-light as in
Lemma~\ref{lem:3light}. 

Instead, we now improve the approximation ratio of
$3+\eps$ (from Theorem~\ref{thm:kcest-k+eps}) to $5/3+\eps$.
We re-use some ideas of the approach for \cest{2}.

\begin{figure}[tb]
  \centering
  \includegraphics[scale=1.5]{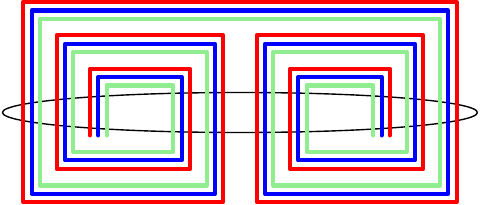}
    \caption{A difficult portal crossing of a \cest{3} instance.}
    \label{fig:spiral}
\end{figure}

  To this end, take an optimal solution $T$ for \cest{3}.  The terminals are red,
  green, and blue; we call the corresponding trees $R$, $G$, and $B$.
  We assume that $B$ is the cheapest among the three trees.  Now, we
  construct a quadtree partitioning the plane and choose the portals,
  for a given~$\eps$, as described in Section~\ref{sec:2cest}.  We
  then make the solution portal-respecting, which
  results in a solution $T^*$ consisting
  of trees $R^*$, $G^*$, and $B^*$. In expectation, this increases the length of
  each of the trees (and hence, of~$T$) by a factor of at most $1+\eps$.

  First, we show that we need few portal passes if the blue and the green
  tree do not \emph{meet} at any portal, that is, no blue and green
  passes are adjacent.

\begin{lemma}
  \label{lm:7passes}
  Consider a portal-respecting solution $T^*$ to \cest{3} consisting of
  trees $R^*$, $G^*$, $B^*$.  If $B^*$ and $G^*$ do not meet at
  any portal, then $T^*$ can be transformed into a 7-light
  portal-respecting solution.
\end{lemma}

\begin{proof}
Apply the portal-crossing reduction from Lemma~\ref{lem:3light} and
consider a portal.
Recall that, after this operation, there are no
$rbrb$ and $rgrg$ subsequences in the passes of the portal.
Here, $r$, $b$, and $g$ correspond to the passes of the trees $R^*$,
$B^*$, and $G^*$, respectively.
If the portal has only one blue or one green pass,
then the solution is already 7-light at the portal (with the longest
possible sequences $rgrbrgr$ and $rbrgrbr$, respectively).
Otherwise, it contains at least two blue and at least two green passes. Notice that the sequence of
passes must be $r$-alternate, that is, of the form $\dots r\circ r\circ r\dots$
since blue and green do not meet.
Thus, a sequence of more than 7 passes must contain a subsequence
$grbrgrb$ (or a symmetric one, $brgrbrg$). These subsequences are reducible. See
Figure~\ref{fig:bluegreen}
for one of the possible cases, the other cases are analogous.
\end{proof}

Now, we show that $T^*$ can be transformed into a 10-light
portal-respecting solution~$T'$ of length at most
$|R^*|+|G^*|+3|B^*|$.

\begin{lemma}
  \label{lm:3reduction}
  A portal-respecting solution $T^*$ to \cest{3}, consisting of trees
  $B^*$, $R^*$, and $G^*$,
  can be transformed into a portal-respecting solution $T'$ such that
  \begin{enumerate}[(i)]
    \item $T'$ passes at most $10$ times through each portal, and
    \item $|T'| \le |R^*|+|G^*|+3|B^*|$. 
  \end{enumerate}
\end{lemma}

\begin{proof}
  We define a \emph{$\bg$-solution}; informally, this is a solution in
    which we are allowed to connect green branches to the blue tree.
    Formally, a $\bg$-solution is a set of three non-intersecting curves
  spanning the terminals of the corresponding color such that the blue and the
  red curves are connected (that is, they are trees) and every
  green component is attached to the blue tree. Note that, in general, a $\bg$-solution is
  not a valid \cest{3} solution.
  We prove the lemma in two steps. First, we show that $T$ can be
  transformed to a portal-respecting $\bg$-solution $T^{\bg}$ with at most $6$ passes per
  portal having the same (or smaller) length. Then, we show how $T^{\bg}$ can be further
  modified to get a portal-respecting solution $T^*$ with at most $10$ passes per
  portal and the desired length.

  In order to construct $T^{\bg}$ from $T^*$,
  we first replace, as in Lemma~\ref{lem:3light},
  all uni-colored sequences
  of passes (that is, consisting of passes of the same color) in a portal by a single pass, and all bi-colored
  sequences of passes by at most 3 passes.
  We call this procedure the \emph{initial reduction}.
  Consequently, uni-colored and bi-colored
  portals have at most 3 passes; hence, we focus on the portals containing passes of all
  three colors.

  \begin{figure}[tb]
    \centering
    \begin{subfigure}[b]{.3\textwidth}
      \centering
      \includegraphics[page=1]{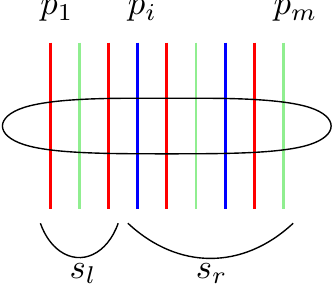}
      \caption{a tri-colored portal after the initial reduction}
      \label{fig:portal53a}
    \end{subfigure}
    \hfil
    \begin{subfigure}[b]{.3\textwidth}
      \centering
      \includegraphics[page=2]{portal53}
      \caption{eliminating the green passes that meet a blue pass}
      \label{fig:portal53elim}
    \end{subfigure}
    \hfil
    \begin{subfigure}[b]{.3\textwidth}
      \centering
      \includegraphics[page=3]{portal53}
      \caption{applying the portal reduction to $s_l$ and $s_r$}
      \label{fig:portal53b}
    \end{subfigure}
    \caption{Construction of a $\bg$-solution to a portal.}
    \label{fig:portal53}
  \end{figure}

  We can assume that there is a portal in which a blue and a green pass
  are adjacent (otherwise, we already have a 7-light instance by
  Lemma~\ref{lm:7passes}.).  We eliminate the
  green pass by connecting it to its neighboring blue pass; see Figure~\ref{fig:portal53elim}. We thus may consider
  the blue and the green tree to be connected and together to form a blue-green tree.
  Its length is $|B^*|+|G^*|$.
  After this step, the sequence of passes is R-alternate.

  Consider now a tri-colored portal in $T$ after the initial reduction.
  Let $(p_1, \dots, p_m)$ be a sequence of passes through the portal and
  suppose that pass $p_i$ ($1\le i\le m$) is the leftmost blue pass; see
  Figure~\ref{fig:portal53a}. Split the sequence into two subsequences
  $s_l = (p_1,\dots, p_{i-1})$ and $s_r = (p_i,\dots, p_m)$. Since~$p_i$ is the
  leftmost blue pass, $s_l$ is bi-colored and, by the initial reduction,
  $|s_l| \le 3$. Regarding~$s_r$, we apply the portal-crossing reduction according to
  Lemma~\ref{lem:3light}, by viewing the blue and the green passes as passes of a
  single blue-green tree; see Figure~\ref{fig:portal53b}. As a result, we get a new
  instance~$s_r'$
  with $|s_r'| \le 3$. Recall that, before this step, the blue and the
  green tree are connected; thus, every blue and green pass is connected to the
  left-most pass~$p_i$. Further, after disconnecting a blue pass in the proof of
  Lemma~\ref{lem:3light}, one part remains connected to~$p_i$ by a blue path,
  and the other part gets connected to~$p_i$ by a blue segment. For the green
  passes, one part remains connected to~$p_i$, while the other part gets connected
  to~$p_i$ by a green segment.
  Since the blue segments where connected before this step, the blue tree
  remains connected after the portal-crossing reduction. The green tree is split into
  subtrees that are connected to the blue tree. Therefore, the new instance is a
  $\bg$-solution. The sequence of passes in the portal for $T^{\bg}$ is a
  concatenation of $s_l$ and $s_r'$ and, hence, has at most 6 passes per portal.

  Note that the constructed solution $T^{\bg}$ has at most 2 blue passes per portal. We add
  a green shell to~$B$ to connect green branches. This increases the number of passes
  per portal by at most~4. The resulting solution $T^*$ has
  length bounded by $|R^*|+|G^*|+3|B^*|$ and at most 10 passes per portal.
\end{proof}

\begin{figure}[tb]
  \centering
  \includegraphics[page=1]{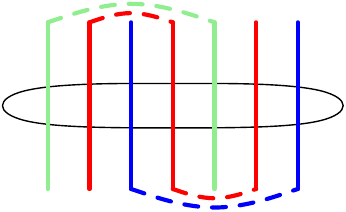}
  \hfil
  \includegraphics[page=2]{bluegreen}
  \caption{Constructing a 7-light solution to an instance without 
    adjacent blue-green passes (one of several possible cases).}
  \label{fig:bluegreen}
\end{figure}

Before we describe our approximation algorithm, we first need to discuss the perturbation step.
The perturbation itself is the same as in Section~\ref{sec:2cest}:
we move each terminal to a uniquely defined closest grid point (we assign the
grid points of even row and odd column to the third color) and merge
terminals of the same color to one terminal.
However, we need a different technique to transform a solution to the 
original instance into a solution to the perturbed instance and vice versa.

\begin{lemma}\label{lem:3cestPerturbed}
  Let $I$ be a \cest{3} instance with $n$ terminals, let $\opt_I$ be
  the length of an optimal solution to~$I$, and let $\eps>0$. Then,
  we can place an $(L\times L)$-grid with $L = O(n/\eps)$ such
  that, for the perturbed instance~$I^*$ of~$I$, $\opt_{I^*} \le
  (1+\eps)\opt_I$.
\end{lemma}
\begin{proof}
  We proceed similar to the proof of Lemma~\ref{lem:perturbed} by connecting each
  terminal of~$I^*$ to the nearest point of its corresponding tree. Since this
  connection can cross segments of two colors, we have to be more careful with
  the rerouting. We choose~$L$ as a power of~$2$ within the 
  interval~$[7\sqrt{2}n/\eps,14\sqrt{2}n/\eps]$.
  
  Fix an optimal solution to~$I$. Consider a terminal~$v$ of~$I^*$ of,
  say, green color. Connect it to the nearest point of the green tree by a
  straight-line segment~$s$. Note that the length~$|s|$ of this segment is
  bounded by the distance of~$v$ to its corresponding terminal in~$I$, which
  is at most~$\sqrt{2}$ as we are assuming a unit grid.

  Assume that~$s$ intersects~$\sigma$ red or blue segments.  We order
  the segments by non-increasing
  distance between~$v$ and their intersection point with~$s$; see Figure~\ref{fig:3cestperturbed} for an example.
  We reroute the first three segments along~$s$ going around~$v$.
  This yields a 3-layer shell around~$v$.  Consider the next
  segment according to the ordering and reroute it along~$s$ and around~$v$.
  We can view the crossing point of this segment with the shell as a portal
  on one side of~$s$. This portal contains four bi-colored passes. Using
  Lemma~\ref{lem:3light}, we reduce the number of passes to at most three.
  Now, we stretch this portal around~$v$ along~$s$ until it reaches the crossing
  point on the other side of~$s$. Since the portal has at most three passes, the
  shell around~$v$ still consists of three layers. We repeat this until all
  segments are rerouted around~$v$.

  By using this rerouting for every terminal in~$I^*$, the total
  length of the solution increases by at most~$(\sqrt{2}+6\cdot \sqrt{2})n=7\sqrt{2}n$.
  Since $\opt_I\ge L$, the
  length of an optimal solution to~$I^*$ is at most
  \[\opt_I+7\sqrt{2}n\le (1+\eps)\opt_I~.\qedhere\]
\end{proof}

\begin{figure}
  \centering
  \begin{subfigure}[b]{.47\textwidth}
    \centering
    \includegraphics[page=1]{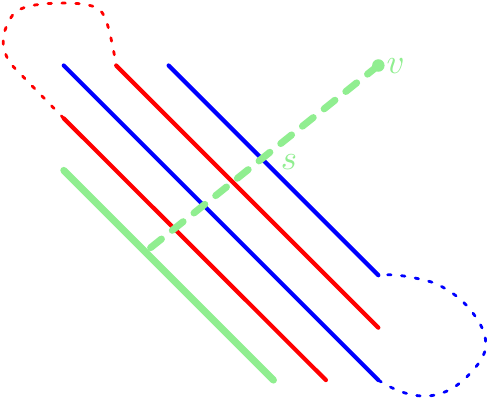}
    \caption{$s$ intersects four red/blue segments}
  \end{subfigure}
  \hfil
  \begin{subfigure}[b]{.47\textwidth}
    \centering
    \includegraphics[page=2]{3cestperturbed}
    \caption{adding the shell for the first three segments}
  \end{subfigure}
  
  \medskip
  
  \begin{subfigure}[b]{.47\textwidth}
    \centering
    \includegraphics[page=4]{3cestperturbed}
    \caption{adding the shell and portal for the fourth segment}
  \end{subfigure}
  \hfil
  \begin{subfigure}[b]{.47\textwidth}
    \centering
    \includegraphics[page=7]{3cestperturbed}
    \caption{resolving the portal and extending it to the other side}
  \end{subfigure}
  
  \medskip
  
  \begin{subfigure}[b]{.47\textwidth}
    \centering
    \includegraphics[page=8]{3cestperturbed}
    \caption{resulting solution}
  \end{subfigure}
  \caption{An example for the rerouting in the proof of Lemma~\ref{lem:3cestPerturbed}.}
  \label{fig:3cestperturbed}
\end{figure}

Analogously to the proof of Lemma~\ref{lem:3cestPerturbed},
we transform a solution to a perturbed instance back into one to the original
instance not increasing the length by much. Then, we combine the lemmas to
prove the main result of this section.

\begin{lemma}\label{lem:3cestPerturbedBack}
    We can transform a solution~$\mathcal T$ to the perturbed instance~$I^*$
    into a solution to the original
    instance~$I$, increasing the length by at most $\eps\opt_I$.
\end{lemma}
\begin{proof}
  Iteratively connect each terminal of the original instance to the
  solution $\mathcal T$ analogously to the proof of
  Lemma~\ref{lem:3cestPerturbed}.  Again, we pay at most 14 units per
  terminal, which yields the claim.
\end{proof}

Using Lemmas~\ref{lm:3reduction},~\ref{lem:3cestPerturbed}
and~\ref{lem:3cestPerturbedBack}, we are now ready to prove
the main result of this section.

\begin{theorem}
  \label{thm:3cest}
  For every $\eps>0$, \cest{3} admits a $(5/3+\eps)$-approximation
  algorithm.
\end{theorem}

\begin{proof}
  Let~$\eps'=\sqrt[3]{1+3\eps/5}-1$. Let~$T$ be an optimal solution to
  a \cest{3} instance~$I$ with trees $R$, $G$ and~$B$. W.l.o.g., assume
  that~$|B|\le |R|,|G|$. Denote by~$\opt_I=|R|+|G|+|B|$ the length of~$T$.
  We first construct a portal-respecting solution~$T^*$
  of length
  \[|T^*|=|R^*|+|G^*|+|B^*|\le (1+\eps')(|R|+|G|+|B|)~.\]
  Then, Lemma~\ref{lm:3reduction} yields
  an optimal 10-light portal-respecting solution~$T'$ of length
\begin{flalign*}
  |T'| \le |R^*|+|G^*|+3|B^*| \le 5/3 \cdot |T^*| & \le 5/3 \cdot (1+\eps') \cdot (|R|+|G|+|B|) \\
  & = 5/3 \cdot (1+\eps') \cdot \opt_I~.
\end{flalign*}

Using a dynamic program similar to the one described in
Section~\ref{sec:2cest} and Lemma~\ref{lem:3cestPerturbed}, we find a
10-light portal-respecting solution
of length $(1+\eps')|T'|$ to the perturbed instance~$I^*$ of~$I$.
By Lemma~\ref{lem:3cestPerturbedBack},
we can transform our solution to~$I^*$ into a solution to~$I$
whose total length is bounded by
\[(1+\eps')^2|T'|
\le 5/3(1+\eps')^3\opt_I
< (5/3+\eps)\opt_I~. \qedhere\]
\end{proof}

\section{Conclusion}

We have presented approximation algorithms for \kcest.  
We leave the following questions open.
For $k=2$, we achieved a PTAS, for $k=3$, a ratio of $5/3+\eps$, and
for general~$k$, ratios $k + \eps$ and $O(\sqrt n \log k)$.

Apart from improving approximation algorithms for \kcest with $k\ge 3$,
a number of interesting questions remain open.
Is \kcest APX-hard for some $k \ge 3$? 
Can we improve the running time of the PTAS for
\cest{2} from $O(n^{O(1/\eps)})$ to $O(n (\log n)^{O(1/\eps)})$ as
Arora~\cite{Aro98} did for EST?

Currently, we are studying an ``anchored'' version of \kcest where the
only allowed Steiner points are input points of a different color.
Any $\alpha$-approximation for \kcest yields an $\alpha(1 +
\sqrt{3})/2$- approximation for the anchored version.

\paragraph*{Acknowledgments.}
We are grateful to Alon Efrat, Jackson Toeniskoetter, and Thomas van Dijk for
the initial discussion of the problem.

\bibliographystyle{alpha}
\bibliography{abbrv,paper}

\newcommand{\etalchar}[1]{$^{#1}$}
\begin{thebibliography}{RBvK{\etalchar{+}}08}

\bibitem[Aro98]{Aro98}
Sanjeev Arora.
\newblock Polynomial time approximation schemes for {E}uclidean traveling
  salesman and other geometric problems.
\newblock {\em J. ACM}, 45(5):753--782, 1998.

\bibitem[ARRC11]{ARRC11}
Basak Alper, Nathalie~Henry Riche, Gonzalo Ramos, and Mary Czerwinski.
\newblock Design study of {LineSets}, a novel set visualization technique.
\newblock {\em IEEE Trans. Vis. Comput. Graphics}, 17(12):2259--2267, 2011.

\bibitem[Bes03]{Bes03}
Sergey Bespamyatnikh.
\newblock Computing homotopic shortest paths in the plane.
\newblock {\em J. Algorithms}, 49(2):284--303, 2003.

\bibitem[BF98]{BF98}
Oliver Bastert and Sandor~P. Fekete.
\newblock Geometric wire routing.
\newblock Technical Report 96.247, Universit{\"a}t zu K{\"o}ln, 1998.
\newblock Available at http://e-archive.informatik.uni-koeln.de/247.

\bibitem[BKM09]{BKM09}
Glencora Borradaile, Philip Klein, and Claire Mathieu.
\newblock An ${O}(n \log n)$ approximation scheme for {Steiner} tree in planar
  graphs.
\newblock {\em ACM Trans. Algorithms}, 5(3):31, 2009.

\bibitem[CG85]{CG85}
F.~R.~K. Chung and R.~L. Graham.
\newblock A new bound for {Euclidean} {Steiner} minimal trees.
\newblock {\em Annals New York Acad. Sci.}, 440(1):328--346, 1985.

\bibitem[CHKL13]{CHKL13}
Timothy~M. Chan, Hella-Franziska Hoffmann, Stephen Kiazyk, and Anna Lubiw.
\newblock Minimum length embedding of planar graphs at fixed vertex locations.
\newblock In Stephen Wismath and Alexander Wolff, editors, {\em Graph Drawing
  (GD'13)}, volume 8242 of {\em LNCS}, pages 376--387. Springer, Heidelberg,
  2013.

\bibitem[CPC09]{CPC09}
Christopher Collins, Gerald Penn, and Sheelagh Carpendale.
\newblock Bubble sets: Revealing set relations with isocontours over existing
  visualizations.
\newblock {\em IEEE Trans. Vis. Comput. Graphics}, 15(6):1009--1016, 2009.

\bibitem[EHKP14]{EHKP14}
Alon Efrat, Yifan Hu, Stephen~G. Kobourov, and Sergey Pupyrev.
\newblock Mapsets: Visualizing embedded and clustered graphs.
\newblock In Christian Duncan and Antonios Symvonis, editors, {\em Graph
  Drawing (GD'14)}, volume 8871 of {\em LNCS}, pages 452--463. Springer,
  Heidelberg, 2014.

\bibitem[EKL06]{EKL06}
Alon Efrat, Stephen~G. Kobourov, and Anna Lubiw.
\newblock Computing homotopic shortest paths efficiently.
\newblock {\em Comput. Geom. Theory Appl.}, 35(3):162--172, 2006.

\bibitem[EN11]{EN11}
Jeff Erickson and Amir Nayyeri.
\newblock Shortest non-crossing walks in the plane.
\newblock In {\em Proc. ACM-SIAM Symp. Discrete Algorithms (SODA'11)}, pages
  297--308, 2011.

\bibitem[GJ79]{GJ79}
Michael~R. Garey and David~S. Johnson.
\newblock {\em Computers and Intractability: A Guide to the Theory of
  {NP}-Completeness}.
\newblock W. H. Freeman \& Co., New York, NY, USA, 1979.

\bibitem[HKK{\etalchar{+}}13]{HKKLS+13}
Ferran Hurtado, Matias Korman, Marc Kreveld, Maarten L{\"o}ffler, Vera
  Sacrist\'{a}n, Rodrigo~I. Silveira, and Bettina Speckmann.
\newblock Colored spanning graphs for set visualization.
\newblock In Stephen Wismath and Alexander Wolff, editors, {\em Graph Drawing
  (GD'13)}, volume 8242 of {\em LNCS}, pages 280--291. Springer, Heidelberg,
  2013.

\bibitem[KMN01]{KMN01}
Yoshiyuki Kusakari, Daisuke Masubuchi, and Takao Nishizeki.
\newblock Finding a noncrossing {S}teiner forest in plane graphs under a 2-face
  condition.
\newblock {\em J. Combin. Optim.}, 5:249--266, 2001.

\bibitem[LMM{\etalchar{+}}95]{LMMPS95}
Thomas~M. Liebling, Fran{\c{c}}ois Margot, Didier M{\"u}ller, Alain Prodon, and
  Lynn Stauffer.
\newblock Disjoint paths in the plane.
\newblock {\em ORSA J. Comput.}, 7(1):84--88, 1995.

\bibitem[L{\"o}f11]{Lof11}
Maarten L{\"o}ffler.
\newblock Existence and computation of tours through imprecise points.
\newblock {\em Int. J. Comput. Geom. Appl.}, 21(1):1--24, 2011.

\bibitem[MHT10]{MT10}
Matthias M{\"u}ller-Hannemann and Siamak Tazari.
\newblock A near linear time approximation scheme for {S}teiner tree among
  obstacles in the plane.
\newblock {\em Comput. Geom. Theory Appl.}, 43(4):395--409, 2010.

\bibitem[Mit99]{Mit99a}
Joseph~S.B. Mitchell.
\newblock Guillotine subdivisions approximate polygonal subdivisions: {A}
  simple polynomial-time approximation scheme for geometric {TSP}, $k$-{MST},
  and related problems.
\newblock {\em {SIAM} J. Comput.}, 28(4):1298--1309, 1999.

\bibitem[Mit00]{m-gspno-HBCG00}
Joseph~S.B. Mitchell.
\newblock Geometric shortest paths and network optimization.
\newblock In J.~Urrutia and J.-R. Sack, editors, {\em Handbook of Computational
  Geometry}, chapter~15, pages 633--701. North-Holland, 2000.

\bibitem[Pap99]{Pap99}
Evanthia Papadopoulou.
\newblock $k$-pairs non-crossing shortest paths in a simple polygon.
\newblock {\em Int. J. Comput. Geom. Appl.}, 9(6):533--552, 1999.

\bibitem[PM05]{PB05}
Valentin Polishchuk and Joseph S.~B. Mitchell.
\newblock Touring convex bodies -- {A} conic programming solution.
\newblock In {\em Canadian Conf. Comput. Geom.}, pages 290--293, 2005.

\bibitem[PM07]{PM07}
Valentin Polishchuk and Joseph S.~B. Mitchell.
\newblock Thick non-crossing paths and minimum-cost flows in polygonal domains.
\newblock In {\em Proc. ACM Symp. Comput. Geom. (SoCG'07)}, pages 56--65, 2007.

\bibitem[Raz08]{Raz08}
Mina Razaghpour.
\newblock The {Steiner} ratio for the obstacle-avoiding {Steiner} tree problem.
\newblock Master's thesis, University of Waterloo, 2008.
\newblock Available at http://hdl.handle.net/10012/4055.

\bibitem[RBvK{\etalchar{+}}08]{RBKMSW08}
Iris Reinbacher, Marc Benkert, Marc~J. van Kreveld, Joseph S.~B. Mitchell, Jack
  Snoeyink, and Alexander Wolff.
\newblock Delineating boundaries for imprecise regions.
\newblock {\em Algorithmica}, 50(3):386--414, 2008.

\bibitem[TSN96]{TSN96}
J.~Takahashi, H.~Suzuki, and T.~Nishizeki.
\newblock Shortest noncrossing paths in plane graphs.
\newblock {\em Algorithmica}, 16:339--357, 1996.

\bibitem[Ver13]{v-hcor-GD12}
Kevin Verbeek.
\newblock Homotopic $\mathcal{C}$-oriented routing.
\newblock In Walter Didimo and Maurizio Patrignani, editors, {\em Graph Drawing
  (GD'12)}, volume 7704 of {\em LNCS}, pages 272--278. Springer, Heidelberg,
  2013.

\end{thebibliography}

\end{document}